\documentclass[a4paper,10pt, review]{elsarticle}

\usepackage{lineno}
\modulolinenumbers[10]

\journal{ = = = = = = = =}

\usepackage{amsmath,amssymb,amsthm}
\usepackage{graphicx}
\usepackage{epstopdf}
\usepackage[usenames,dvipsnames,table]{xcolor}
\usepackage{hhline}
\usepackage[normalem]{ulem}

\usepackage{subfigure}
\usepackage{mathtools}
\usepackage{wrapfig}
\usepackage{lipsum}
\usepackage{mathrsfs}
\usepackage[T1]{fontenc}
\usepackage{algorithm}
\usepackage{algorithmic}
\usepackage[utf8]{inputenc}

\definecolor{highlight-col}{RGB}{0,125,0}
\newcommand{\colb}[1]{{\color{highlight-col}{\textbf{\textit{#1}}}}}

\newcommand{\probname}[1]{{\color{orange!80!black}{\textbf{{#1}}}}}

\usepackage{pbox}
\usepackage{array}
\newcolumntype{P}[1]{>{\centering\arraybackslash}p{#1}}
\newcolumntype{M}[1]{>{\centering\arraybackslash}m{#1}}

\newtheorem{observation}{Observation}

\newtheorem{theorem}{Theorem}
\newtheorem{lemma}[theorem]{Lemma}
\usepackage{amssymb}

\usepackage[backref=page]{hyperref}
\hypersetup{
    bookmarks=true,         
    unicode=false,          
    pdftoolbar=true,        
    pdfmenubar=true,        
    pdffitwindow=false,     
    pdfstartview={FitH},    
    pdftitle={My title},    
    pdfauthor={Author},     
    pdfsubject={Subject},   
    pdfcreator={Creator},   
    pdfproducer={Producer}, 
    pdfkeywords={keyword1, key2, key3}, 
    pdfnewwindow=true,      
    colorlinks=true,       
    linkcolor=red,          
    citecolor=green!50!black,        
    filecolor=magenta,      
    urlcolor=cyan           
}

\newcommand{\np}{$\mathsf{NP}$}

\newcommand{\ptas}{$\mathsf{PTAS}$}
\newcommand{\cz}{$\mathcal{Z~}$} 

\usepackage{pbox}
\usepackage{array}
\newcolumntype{P}[1]{>{\centering\arraybackslash}p{#1}}
\newcolumntype{M}[1]{>{\centering\arraybackslash}m{#1}}

\newcounter{word}
\makeatletter
\newcommand*{\LBL}{%
  \@dblarg\@LBL
}
\def\@LBL[#1]#2{%
  \begingroup
    \renewcommand*{\theword}{#2}%
    \refstepcounter{word}%
    \label{#1}%
    #2%
  \endgroup
}

\usepackage[framemethod=tikz]{mdframed}
\usepackage{lipsum}

\usepackage[framemethod=tikz]{mdframed}
\newmdenv[
  backgroundcolor=gray!15,
  topline=false,
  bottomline=false,
  rightline=false,
  skipabove=\topsep,
  skipbelow=\topsep
]{siderules}

\usepackage{calc}
\newlength{\maxmin}
\definecolor{tab-color}{RGB}{170,80,20}
\newcommand{\tabcol}[1]{\color{tab-color}{\textbf{#1}}}
\newcommand{\tabcolm}[1]{\color{tab-color}{\boldmath{#1}}}

\usepackage{wrapfig}

\begin{document}

\begin{frontmatter}

\title{Faster Approximation for Maximum Independent Set on Unit Disk Graph}

\author{Subhas C. Nandy}
\ead {nandysc@isical.ac.in}

\author{Supantha Pandit\corref{cor}}
\cortext[cor]{Corresponding author}
\ead{pantha.pandit@gmail.com}

\author{Sasanka Roy\corref{cor1}}
\ead{sasanka@isical.ac.in}

\address{Indian Statistical Institute, Kolkata, India}

\begin{abstract}
Maximum independent set from a given set $D$ of unit disks intersecting a horizontal line can be solved in $O(n^2)$ time and $O(n^2)$ space. As a corollary, we design a factor 2 approximation algorithm for the maximum independent set problem on unit disk graph which takes both time and space of $O(n^2)$. The best known factor 2 approximation algorithm for this problem runs in $O(n^2 \log n)$ time and takes $O(n^2)$ space \cite{Jallu2016,Das2016}.

\end{abstract}


\end{frontmatter}

\linenumbers

\section{Introduction}
\label{intro}

\colb{Intersection graphs} of geometric objects have used to model several problems that arise in real scenarios \cite{Roberts1978}. Two important applications of these graphs are frequency assignment in cellular networks \cite{Hale1980,Malesinska1997} and map labeling \cite{Agarwal1998}. If the geometric objects are disks then the corresponding intersection graph $G(V,E)$ is called the \colb{disk graph}. Here the vertex set $V$ corresponds to a given set of disks in the plane, and there is an edge between two vertices in $V$ iff the corresponding two disks intersect. 

A \colb{unit disk graph} is an intersection graph where each disk is of diameter 1. Let $G(V,E)$ be a given graph. A set $V'\subseteq V$ is said to be an \colb{independent set} of $G$ if no two vertices in $V'$ are connected by an edge in $G$. In the \colb{maximum independent set (MIS)} the goal is to find an independent set $V'$ which has the maximum cardinality. In this paper, we consider the following problem.

\begin{siderules}
\noindent\probname{Maximum Independent Set on Unit Disk Graph ({\textit{MISUDG}}):} Given a unit disk graph $G(V,E)$, find an independent set of $G$ whose cardinality is maximum.
\end{siderules}

To provide an approximation algorithm for {\it MISUDG}, we consider the following problem.

\begin{siderules}
\noindent \probname{{\textit{MISUDG-L}}:}\label{prob1} Given a set $D_i$ of $n_i$ unit disks that are intersected by horizontal line $L_i$, find a subset $D'\subseteq D_i$ of maximum cardinality such that no two disks in $D'$ have a common intersection point.
\end{siderules}

\vspace{.3cm}
\noindent \textbf{Related Work:} The {\it MISUDG} problem is known to be \np-complete \cite{Wang1988,Garey1979,Clark1990}.  In Table \ref{table-comparison}, we demonstrate a comparison study of the progress on {\it MISUDG}.

\begin{table}[htbp]
\begin{center}
\setlength{\tabcolsep}{12pt}
\begin{tabular}{|l|c|c|M{1.5cm}|}
\hline
Reference   & Factor  & Time & Space\\
\hline 

Marathe et al. \cite{Marathe1995} & 3 & $O(n^2)$ &$O(n)$\\
\hline
Das et al. \cite{Das2015}   & 2 & $O(n^3)$ &$O(n^2)$\\
\hline
Jallu and Das \cite{Jallu2016}  &2 &$O(n^2 \log n)$&$O(n^2)$\\
\hline
 Das et al. \cite{Das2016}   & 2.16 & $O(n\log^2n)$ &$O(n\log n)$\\
\hline
{\tabcol{Theorem \ref{theo-n2-fac2}}} & {\tabcol{2}}  &{\tabcolm{$O(n^2)$}} & {\tabcolm{$O(n^2)$}}\\
\hline
\end{tabular}
\end{center}
\caption{Comparison table}
\label{table-comparison}
\end{table}
%
%

Matsui \cite{Matsui1998} consider the {\it MISUDG} problem. If the disk centers are located inside a strip of fixed height $k$, then this problem can be solved in $O(n^{4\lceil{\frac{2k}{\sqrt{3}}}\rceil})$ time. Further, for any integer $r\geq 2$, Matsui \cite{Matsui1998} provided a $(1-\frac{1}{r})$ factor approximation algorithm for the same problem which takes $O(rn^{4\lceil{\frac{2(r-1)}{\sqrt{3}}}\rceil})$ time and $O(n^{2r})$ space. Das et al. \cite{Das2015}, also designed a \ptas~ for {\it MISUDG} problem by using the {\it shifting strategy} of Hochbaum and Maass \cite{Hochbaum1985}. For a given positive integer $k > 1$, they gave a $(1+\frac{1}{k})^2$ factor approximation algorithm which runs in $O(k^4n^{\sigma_k \log k}+n\log n)$ time and $O(n+k\log k)$ space, where $\sigma_k \leq \frac{7k}{3}+2$. Recently, Jallu and Das \cite{Jallu2016}, improved the running time of the same problem to $n^{O(k)}$ by keeping the approximation factor same. A fixed parameter tractable algorithm for the {\it MISUDG} problem was proposed by van Leeuwen \cite{vanLeeuwen2005}. The running time of that algorithm is $O(t^2 2^{2t} n)$, where the parameter $t$ represents the {\it thickness}\footnote{A {\it UDG} is said to have thickness $t$, if each strip in the slab decomposition of width 1 of the {\it UDG} contains at most $t$ disk centers} of the {\it UDG}. 

\vspace{.3cm}
\noindent \textbf{Our Contributions:} 
\begin{itemize}

\item We design an exact algorithm for {\it MISUDG-L} problem which runs in $O(n^2)$ time using $O(n^2)$ space.

\item We design a factor 2 approximation algorithm for {\it MISUDG} problem which takes both $O(n^2)$ time and space. It is an improvement over the best known result on this problem proposed by Jallu et al. \cite{Jallu2016}. They gave a factor 2 approximation algorithm for  this problem where the time and space complexities are $O(n^2 \log n)$  and $O(n^2)$ respectively.



\end{itemize}

\vspace{.3cm}
\noindent\textbf{Notations and Definitions:}  Let $D=\{d_1,d_2,\ldots,d_n\}$ be a set of $n$ unit disks in the plane. The center of the disk $d_i \in D$ is $c_i$. The $x$-coordinate of $c_i$ is $x(c_i)$. For a given set $S$ of disks, $|S|$ is the cardinality of $S$. The line segment connecting two points $s$ and $t$ is denoted by $\overline{st}$.

\section{\boldmath{$O(n^2)$} time exact algorithm for {\textit{MISUDG-L}} problem} \label{sec-prob-a}

In this section, we design an exact dynamic programming based algorithm for {\it MISUDG-L} problem. Let $D_i=\{d_1, d_2, \ldots d_{n_i}\}$ be a set of $n_i$ unit disks intersecting a horizontal line $L_i$. We partition the set $D_i$ into two sets $D_i^a$ and $D_i^b$, where $D_i^a$ is the set of all disks in $D_i$ whose centers are above the horizontal line $L_i$ and $D_i^b$ is the set of all disks in $D_i$ whose centers are below the horizontal line $L_i$. To design the dynamic programming algorithm, we need the following two lemmas.

\begin{lemma}\label{upper-side-ind-lemma} Let $d_1, d_2, d_3 \in D_i^a$ be three disks with centers $c_1$, $c_2$, and $c_3$ respectively. Assume that $x(c_1) < x(c_2) < x(c_3)$. Now if $d_1$, $d_2$ and $d_2$, $d_3$ are non-intersecting, then $d_1$, $d_3$ are non-intersecting.
\end{lemma}
\begin{proof} Suppose on the contrary, we assume that $d_1$ and $d_3$ are intersecting. Then clearly the line segment $\overline{c_1c_3}$ must be fully covered by $d_1$ and $d_3$. Since $x(c_1) < x(c_2) < x(c_3)$, $c_2$ can not be above $\overline{c_1c_3}$. Otherwise, it must intersect $\overline{c_1c_3}$ and hence intersect either $d_1$ or $d_2$. Further, the perpendicular distance between the horizontal line $L_i$ and any point on $\overline{c_1c_3}$ is at most 1. 
\begin{figure}[htbp]
\begin{center}
\includegraphics[scale=.5]{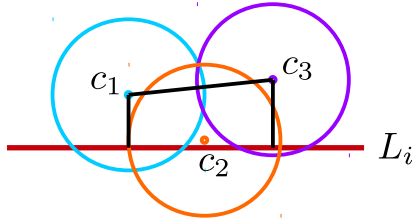} 
\end{center}
\caption{Proof of Lemma \ref{upper-side-ind-lemma}.}
\label{fig-upperside}
\end{figure}
Then, if $c_2$ is below $\overline{c_1c_3}$, it must intersect $\overline{c_1c_3}$ as the centers are above the horizontal line $L_i$. Therefore, we have arrived at a contradiction that either $d_1$, $d_2$ are intersecting or $d_2$, $d_3$ are intersecting.
\end{proof}

\begin{lemma}\label{lower-side-ind-lemma} Let $d_1, d_2 \in D_i^b$  and $d_3 \in D_i^a$ be three disks with centers $c_1$, $c_2$, and $c_3$ respectively. Assume that $x(c_1) < x(c_2) < x(c_3)$. Now if $d_1$, $d_2$ and $d_2$, $d_3$ are non-intersecting, then $d_1$, $d_3$ are non-intersecting.
\end{lemma}

\begin{proof} Suppose on the contrary, we assume that $d_1$ and $d_3$ are intersecting. Then clearly $\overline{c_1c_3}$ is at most 1. Also by the assumption, both $\overline{c_1c_2}$ and $\overline{c_2c_3}$ are greater than 1. Let $V_L$ be a vertical line through $c_2$ (see Figure \ref{fig-lowerside}). The two lines $L_i$ and $V_L$ intersect at a point $O$ and partition the space into four quadrants: `$++$', `$+-$', `$--$', and `$-+$'. The point $c_3$ is in `$++$', whereas $c_1$ is in `$--$'. 
\begin{figure}[htbp]
\begin{center}
\includegraphics[scale=.5]{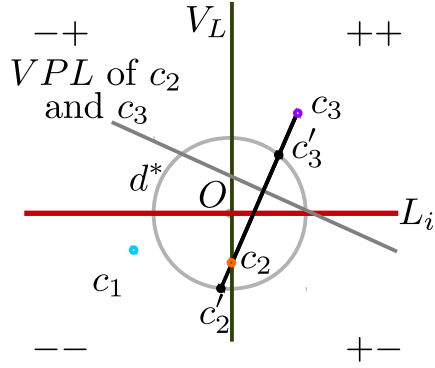} 
\end{center}
\caption{Proof of Lemma \ref{lower-side-ind-lemma}}
\label{fig-lowerside}
\end{figure}
Now consider an unit disk $d^*$ whose center coincides with $O$. Note that, all disks in $D_i$ intersect the line $L_i$. Hence the disk $d_2$ contains the point $O$. Further, since $c_2$ and $c_3$ are non-intersecting, $c_3$ must be outside $d^*$.

Take the segment $\overline{c_2c_3}$ which intersect $d^*$ at $c'_3$. Further, extend the segment $\overline{c_2c_3}$ in the direction of $c_2$ such that it intersect another point $c'_2$ on $d^*$. Consider the segment $\overline{c'_2c'_3}$. Now by an easy observation, we say that, the voronoi partition line ({\it VPL}) of $c'_2$ and $c'_3$ passes through $O$ and intersects the two quadrants `$+-$' and `$-+$'. Again, consider the segment $\overline{c_2c'_3}$. Since $c_2$ is on the line through $c'_2$ and $c'_3$, the slope of the {\it VPL} of $c_2$ and $c'_3$ must be the same as that of $c_2'$ and $c_3'$. Further, this {\it VPL} is to the right of the {\it VPL} of $c'_2$ and $c'_3$ and contains the whole `$--$' quadrant to its left. Due to similar argument, the {\it VPL} of $c_2, c_3$ contains the whole `$--$' quadrant to its left. Since $c_1$ and $c_2$ are in `$--$' quadrant, clearly the point $c_1$ is closer to $c_2$ than $c_3$. Therefore,   $\overline{c_1c_3}$ is greater than 1, since $\overline{c_1c_2}$ is greater than 1. This leads to the contradiction that $\overline{c_1c_3}$ is at most 1. 
\end{proof}

We now describe the algorithm as follows. Let $\{d_1^a, d_2^a, \ldots, d_{n_1}^a\}$ be the set of disks in $D_i^a$ sorted according to their increasing $x$-coordinates. Similarly, let $\{d_1^b, d_2^b, \ldots, d_{n_2}^b\}$ be the set of disks in $D_i^b$ sorted according to their increasing $x$-coordinates. We add two new disks $d_0^a$ and $d_0^b$ which satisfies the following, (i) $d_0^a$ is to the left of $d_1^a$ and $d_0^b$ is to the left of $d_1^b$, (ii) both $d_0^a$ and $d_0^b$ are independent with the disks in $D_i$, and (iii) $d_0^a$ and $d_0^b$ do not intersect each other. For any disk $d \in D_i$ ($d \neq \{d_0^a, d_0^b\}$), define $RI^a(d)$ (resp. $RI^b(d)$) be the rightmost disk in $D_i^a$ (resp. $D_i^b$)  which is independent with $d$ and whose center is to the left of the center of $d$.

We define a subproblem $S(k,\ell)$, for $0\leq k\leq n_1$ and $0\leq \ell\leq n_2$, to be the set of all disks in $D_i^a$ which are to the left of the disk $d_k^a\in D_i^a$ and set of all disks in $D_i^b$ which are to the left of the disk $d_\ell^b\in D_i^b$. Let $I(k,\ell)$ be an optimal set of independent unit disks in $S(k,\ell)$, and let $V(k,\ell)$ be the value of this solution. 

\begin{lemma}\label{lemma-dynamic-proof} Let $D_i^a(k)=\{d_1^a, d_2^a, \ldots, d_k^a\}$ be a set of $k$ leftmost disks in $D_i^a$ and $D_i^b(\ell)=\{d_1^b, d_2^b, \ldots, d_\ell^b\}$ be the set of $\ell$ leftmost disks in $D_i^b$. Now, 
\begin{itemize}
\item[A.] if $x(c_k^a)>x(c_\ell^b)$, then 
\begin{itemize}
\item[(1)] if $d_k^a \in I(k,\ell)$, then $V(k,\ell) = V(RI^a(d_k^a),RI^b(d_k^a)) +1$
\item[(2)] if $d_k^a \notin I(k,\ell)$, then $V(k,\ell) = V(k-1,\ell)$
\end{itemize}
\item[B.] if $x(c_k^a)<x(c_\ell^b)$, then 
\begin{itemize}
\item[(3)] if $d_\ell^b \in I(k,\ell)$, then $V(k,\ell) = V(RI^a(d_\ell^b),RI^b(d_\ell^b)) +1$
\item[(4)] if $d_\ell^b \notin I(k,\ell)$, then $V(k,\ell) = V(k,\ell-1)$
\end{itemize}
\end{itemize}
\end{lemma}

\begin{proof} We prove cases 1 and 2. The proof of the cases 3 and 4 are similar. Here we assume that, $x(c_k^a)>x(c_\ell^b)$, i.e., the disk $d_k^a$ is to the right of the disk $d_\ell^b$. Let $T^*$ be a maximum independent set of disks for subproblem $S(k,\ell)$. There are two possibilities, either  $d_k^a$ is in the optimal solution or not.
\begin{itemize}
\item[$d_k^a\in$] $I(k,\ell)$: Let us assume that, $d_\tau^a = RI^a(d_k^a)$ and $d_\nu^b = RI^b(d_k^a)$. Since, $d_k^a$ is in the optimal solution, no disk in $D_i^a$ (resp. $D_i^b$) whose center is in between the centers of $d_\tau^a$ (resp. $d_\nu^b$) and $d_k^a$ can be present in any feasible solution. Thus any feasible solution contains disks from $D_i^a(\tau)$ and $D_i^b(\nu)$. Therefore, $T^*$ consists of $d_k^a$, together with the optimal solution to the subproblem $S(\tau,\nu)$.

\item[$d_k^a \notin$] $I(k,\ell)$: By an argument similar to case 1, we say that, an optimal solution for $D_i^a(k-1)$ and $D_i^b(\ell)$ gives an optimal solution for $D_i^a$ and $D_i^b$.
\end{itemize}
%
This completes the proof of the lemma.
\end{proof}

Therefore, Lemma \ref{lemma-dynamic-proof} suggests the following recurrence relation:


  \begin{equation*}
  \begin{displaystyle}
  V(k,\ell) = \left. \max
  \begin{cases}
  \left. \begin{aligned}
   & ~~V(RI^a(d_k^a),RI^b(d_k^a))+1, \\
   & ~~V(k-1,\ell),\\
    \end{aligned} ~~~~\right\} \text{for } x(c_k^a)>x(c_\ell^b)\\
    \left.    \begin{aligned}
    & ~~V(RI^a(d_\ell^b),RI^b(d_\ell^b))+1,\\
    & ~~V(k,\ell-1),
    \end{aligned} ~~~~\right\} \text{for } x(c_k^a)<x(c_\ell^b)
  \end{cases}
  \right.
\end{displaystyle}
\end{equation*}

\noindent\textbf{Optimal Solution:} The optimal solution can be found by calling the function $V(n_1,n_2)$ with the base cases $V(k,\ell) =0$ where both $k,\ell =2$. Clearly, the final optimal solution contains the disks $d_0^a$ and $d_0^b$. Hence, we reduce the value of the optimal solution by 2 and remove these two disks from the optimal solution.


\noindent\textbf{Running time:} Let $T(n_i)$ be the total time taken by an algorithm \cz to evaluate $V(n_1,n_2)$. For a particular disk $d\in D_i$, finding either $RI^a(d)$ or $RI^b(d)$ requires $O(n_i)$ time. Hence, in $O(n_i^2)$ time, we find $RI^a(d)$ and $RI^b(d)$ for all $d\in D_i$.  During recursive calls, for a particular disk $d$, the disks $RI^a(d)$ and  $RI^b(d)$ can be found in $O(1)$ time. Therefore, the running time of \cz will be $O(n_i^2)$. Further, this algorithm requires $O(n_i^2)$ space to store the values of $V(k,\ell)$, for $0\leq k \leq n_1$ and $0\leq \ell \leq n_2$. Finally, we now have the following theorem.

\begin{theorem}\label{theo-n2} {\textit{MISUDG-L}} problem can be solved optimally in $O(n_i^2)$ time and $O(n_i^2)$ space.
\end{theorem}

\section{\boldmath{$O(n^2)$} time factor 2 approximation for {\textit{MISUDG}} problem}

In this section, we design a factor 2 approximation algorithm for {\it MISUDG} problem. 
Let $D=\{d_1,d_2,\ldots, d_n\}$ be a set of $n$ unit disks in the plane. We first place horizontal lines from top to bottom with unit distance between each consecutive pair. Assume that there are $k$ such horizontal lines $\{L_1,L_2,\ldots, L_k\}$. Let $D_i \subseteq D$ be the set of disks which are intersected by the line $L_i$. Now we have the following  observation.

\begin{observation}\label{obser-even-odd} Any two disks, $d\in D_i$ and $d'\in D_j$ are independent (non-intersecting) if  $|i-j|>1$, for $1\leq i,j\leq k$.
\end{observation}

Note that, algorithm \cz optimally solves {\textit{MISUDG-L}} problem. Run \cz on each $D_i$, for $1\leq i\leq k$ and let $S_i$ be an independent set of unit disks of maximum cardinality in $D_i$, $1\leq i\leq k$. Let $S_{odd} = \bigcup_{\substack{1\leq i\leq k,\\ i ~\text{is odd}}} S_i$ and $S_{even} = \bigcup_{\substack{1\leq i\leq k,\\ i ~\text{is even}}} S_i$. We set $S$ as $S_{odd}$ or $S_{even}$ depending on whether $|S_{odd}|$ is greater or less than $|S_{even}|$ and report $S$ as the result of our algorithm. We now have the following theorem.

\begin{theorem}\label{theo-n2-fac2} The time and space complexities of our algorithm are both $O(n^2)$ and it produces a result with approximation factor $2$.
\end{theorem}

\begin{proof} Let \textsc{Opt} be an optimal solution for $D$. Form Observation \ref{obser-even-odd}, we say that the disks in $S_{odd}$ are independent, and so $S_{even}$. Also, we have $|S_{odd}|+ |S_{even}| \geq |\textsc{Opt}|$. Therefore, $2|S| = 2\max\{|S_{odd}|, |S_{even}|\}\geq |S_{odd}| + |S_{even}| \geq |\textsc{Opt}|$.  


Since disks in $S_{odd}$ and $S_{even}$ are mutually independent, the total time required for computing $S_{odd}$ or $S_{even}$ is $O(n^2)$. Hence, the total time for reporting $S$ is $O(n^2)$, as required. For each $D_i$, \cz takes $O(n^2)$ space. Hence, the total space complexity is $O(n^2)$.
\end{proof}

\end{document}